\documentclass[conference,a4paper]{IEEEtran}
\usepackage{cite}
\usepackage{graphicx}
\usepackage{psfrag}
\usepackage{subfigure}
\usepackage{url}
\usepackage{amsmath}
\usepackage[normalem]{ulem}
\usepackage{epsfig,colortbl}
\usepackage{amssymb,comment}
\usepackage{enumerate}
\usepackage{times}
\usepackage{multirow,multicol}
\usepackage[ruled, vlined]{algorithm2e}
\usepackage{tabularx}

\newcounter{ctr}

\newtheorem{lemma}{Lemma}
\newtheorem{theorem}{Theorem}

\newtheorem{corollary}{Corollary}

\makeatletter

\newcommand{\be}{\begin{eqnarray}}
\newcommand{\ee}{\end{eqnarray}}
\newcommand{\nn}{\nonumber}

\newcommand{\bm}{\boldmath}

\newcommand{\m}{\mbox{\bm $m$}}

\newcommand{\ccc}{{\mbox{\bm $c$}}}

\newcommand{\eee}{\mbox{\bm $e$}}

\newcommand{\0}{\mbox{\bm  $0$}}

\renewcommand\paragraph{\@startsection{paragraph}{4}{\z@}%
    {1.5ex plus .2ex minus .3ex}%
            {-0em}%
                        {\normalsize\bf}}

\begin{document}
\sloppy
\title{Update-Efficient  Regenerating Codes with Minimum Per-Node Storage}

\author{
   \IEEEauthorblockN{
     Yunghsiang S. Han\IEEEauthorrefmark{1},
     Hong-Ta~Pai\IEEEauthorrefmark{2},
     Rong~Zheng\IEEEauthorrefmark{3}
      and
     Pramod~K.~Varshney\IEEEauthorrefmark{4}}
   \IEEEauthorblockA{
     \IEEEauthorrefmark{1}Dep. of Electrical Eng.
    National Taiwan University of Science and Technology,
    Taipei, Taiwan\\
    Email: yshan@mail.ntust.edu.tw}
   \IEEEauthorblockA{
     \IEEEauthorrefmark{2}Dep. of communication Eng.
   National Taipei University,
    Taipei, Taiwan}
   \IEEEauthorblockA{
     \IEEEauthorrefmark{3}Dep. of Computing and Software,
     McMaster University,
    Hamilton, ON, Canada}
   \IEEEauthorblockA{
     \IEEEauthorrefmark{4}Dep. of EECS,
     Syracuse University,
    Syracuse, USA}
 }

\maketitle

\begin{abstract}
Regenerating codes provide an efficient way to recover data at failed nodes in
distributed storage systems. It has been shown that regenerating codes can be designed
to minimize the per-node storage (called MSR) or minimize the communication
overhead for regeneration (called MBR).  In this work, we propose a new
encoding scheme for $[n,d]$ error-correcting MSR codes that generalizes our
earlier work on error-correcting regenerating codes.  We show that by choosing
a suitable diagonal matrix, any generator matrix of the $[n,\alpha]$
Reed-Solomon (RS) code can be integrated into the encoding matrix.  Hence, MSR
codes with the least update complexity can be found.  An efficient decoding
scheme is also proposed that utilizes the $[n,\alpha]$ RS code to perform data
reconstruction. The proposed decoding scheme has better error correction
capability and incurs the least number of node accesses when errors are present.
\end{abstract}


\section{Introduction}
\label{SEC:Intro}

Cloud storage is gaining popularity as an alternative to enterprise storage
where data is stored in virtualized pools of storage typically hosted by
third-party data centers.  Reliability is a key challenge in the design of
distributed storage systems that provide cloud storage. Both crash-stop and
Byzantine failures (as a result of software bugs and malicious attacks) are
likely to be present during data retrieval. A
crash-stop failure makes a storage node unresponsive to access requests. In
contrast, a Byzantine failure responds to access requests with erroneous data.
To achieve better reliability, one common approach is to replicate
data files on multiple storage nodes in a network. Erasure coding is employed to encode the original
data and then the encoded data  is distributed to storage nodes. Typically, more than one
storage nodes need to be accessed to recover the original data. One popular class
of erasure codes is the maximum-distance-separable (MDS) codes.
With $[n,k]$ MDS codes such as Reed-Solomon (RS) codes, $k$ data items are encoded and
then distributed to and stored at $n$ storage nodes. A user or a data collector
can retrieve the original data  by accessing {\it any} $k$ of the storage
nodes, a process referred to as {\it data reconstruction}.

%

Any storage node can fail due to hardware or software damage.  Data stored at
the failed nodes need to be recovered (regenerated)  to remain functional to perform data reconstruction. The
process to recover the stored (encoded) data at a storage node is called {\it
data regeneration}.
{\it Regenerating codes} first introduced in the pioneer works by Dimakis {\it
et al.} in ~\cite{DIM07,DIM10} allow efficient data regeneration.  To
facilitate data regeneration, each storage node stores $\alpha$ symbols and a
total of $d$ surviving nodes are accessed to retrieve $\beta \le \alpha$
symbols from each node.  A trade-off exists between the storage overhead
and the regeneration (repair) bandwidth needed for data regeneration.  Minimum
Storage Regenerating (MSR) codes first minimize the amount of data stored per
node, and then the repair bandwidth, while Minimum Bandwidth Regenerating (MBR)
codes carry out the minimization in the reverse order. There have been many
works that focus on the design of regenerating
codes~\cite{WU07,WU10,CUL09,WU09,RAS09,PAW11,OGG11,RAS11}. Recently, Rashmi
{\it et al.} proposed optimal exact-regenerating codes  that recover 
the  stored data at the failed node exactly (and thus the name
exact-regenerating)~\cite{RAS11}; however, the authors only consider crash-stop
failures of storage nodes. Han {\it et al.} extended Rashmi's work to construct
error-correcting regenerating codes for exact regeneration that can handle
Byzantine failures~\cite{HAN12-INFOCOM}. In~\cite{HAN12-INFOCOM}, the encoding
and decoding algorithms for both MSR and MBR error-correcting codes were also
provided. In~\cite{RAS12}, the code capability and resilience were discussed
for error-correcting regenerating codes.

In addition to bandwidth efficiency and error correction capability, another
desirable feature for regenerating codes is {\it update
complexity}~\cite{RAW11}, defined as the maximum number of encoded symbols that
must be updated while a single data symbol is modified. Low update complexity
is desirable in scenarios where updates are frequent. Clearly, the update
complexity of a regenerating code is determined by the number of non-zero
elements in the row of the encoding matrix with the maximum Hamming weight. 
The smaller the number, the lower the update complexity is. 

One drawback of the decoding algorithms for MSR codes given
in~\cite{HAN12-INFOCOM} is that, when one or more storage nodes have erroneous
data, the decoder needs to access extra data from many storage nodes (at least
$k$ more nodes) for data reconstruction. Furthermore, when one symbol in the
original data is updated, all storage nodes need to update their respective
data. Thus, the MSR and MBR codes in ~\cite{HAN12-INFOCOM} have the maximum
possible update complexity. Both deficiencies are addressed in this paper.
First, we propose a general encoding scheme for MSR codes. As a special case,
least-update-complexity codes are designed. Second, a new decoding algorithm
is presented.  It not only provides better error correction capability but also
incurs low communication overhead when errors occur in the accessed data. 

%

\section{error-correcting MSR Regenerating Codes}
\label{SEC:review}
In this section, we give a brief overview of data regenerating codes and the
MSR code construction presented in \cite{HAN12-INFOCOM}. 
\subsection{Regenerating Codes}
\label{subSEC:RC}
Let $\alpha$ be the number of symbols stored at each storage node and
$\beta\le\alpha$ the number of symbols downloaded from each storage during
regeneration.  To repair the stored data at the failed node, a helper node
accesses $d$ surviving nodes. The design of regenerating codes ensures that the total  regenerating bandwidth be
much less than that of the original data, $B$. A regenerating code must be capable
of reconstructing the original data symbols and regenerating coded data  at a
failed node.   An $[n,k,d]$ regenerating code requires at least $k$ and $d$
surviving nodes to ensure successful data reconstruction and
regeneration~\cite{RAS11}, respectively, where $n$ is the number of storage
nodes and $k\le d\le n-1$.
%

The cut-set bound given in~\cite{WU07,DIM10} provides a constraint on
the  repair bandwidth. By this bound, any regenerating code must satisfy
the following inequality:
\begin{eqnarray}
B\le \sum_{i=0}^{k-1} \min\{\alpha,(d-i)\beta\}~.\label{main-inequality}
\end{eqnarray}
From~\eqref{main-inequality}, $\alpha$ or $\beta$ can be minimized achieving
either the minimum storage requirement or  the minimum repair bandwidth
requirement, but not both. The two extreme points in~\eqref{main-inequality}
are referred to as the minimum storage regeneration (MSR) and minimum bandwidth
regeneration (MBR) points, respectively.  The values of $\alpha$ and $\beta$
for the MSR point can be obtained by first minimizing $\alpha$ and then minimizing
$\beta$:
\begin{eqnarray}
\alpha&=&d-k+1\nn\\
B&=&k(d-k+1)=k\alpha~,\label{NMSR}
\end{eqnarray}
where we normalize $\beta$ as $1$.\footnote{It has been proved that when designing  $[n,k,d]$ MSR for $k/(n+1)\le 1/2$. it
suffices to consider those with $\beta=1$~\cite{RAS11}.}

There are two categories of approaches to regenerate data at a failed node. If
the replacement data is exactly the same as that previously stored at the
failed node, we call it the {\it exact regeneration}. Otherwise, if the
replacement data only guarantees the correctness of data reconstruction and
regeneration properties, it is called {\it functional regeneration}. In
practice, exact regeneration is more desirable since there is no need to inform
each node in the network regarding the replacement. Furthermore, it is easy to
keep the codes systematic via exact regeneration, where partial data can be
retrieved without accessing all $k$ nodes. The codes designed
in~\cite{RAS11,HAN12-INFOCOM} allow exact regeneration.

\subsection{MSR Regenerating Codes With Error Correction Capability}
\label{subSEC:MSR}
Next, we describe the MSR code construction given in~\cite{HAN12-INFOCOM}. In
the rest of the paper, we assume $d = 2\alpha$. 
The information sequence $\m=[m_0,m_1,\ldots, m_{B-1}]$ can be arranged into
an information vector $U=\left[Z_1Z_2\right]$ with size $\alpha\times d$ such that
 $Z_1$ and $Z_2$ are symmetric matrices with
dimension $\alpha\times\alpha$.
 An $[n,d=2\alpha]$ RS code is adopted to construct the
MSR code~\cite{HAN12-INFOCOM}. Let $a$ be a generator of $GF(2^m)$. In the encoding of the MSR code, we have
\begin{eqnarray}
U\cdot G=C,\label{eq:generator}
\end{eqnarray}
where
 {\footnotesize \begin{equation}
\begin{array}{l}
G=\\
\left[\begin{array}{cccc}
1&1&\cdots&1\\
a^0&a^1&\cdots&a^{n-1}\\(a^0)^2&(a^1)^2&\cdots&(a^{n-1})^2\\
&&\vdots&\\
(a^0)^{\alpha-1}&(a^1)^{\alpha-1}&\cdots&(a^{n-1})^{\alpha-1}\\
(a^0)^\alpha 1&(a^1)^\alpha 1&\cdots&(a^{n-1})^\alpha 1\\
(a^0)^\alpha a^0&(a^1)^\alpha a^1&\cdots&(a^{n-1})^\alpha a^{n-1}\\
(a^0)^\alpha(a^0)^2&(a^1)^\alpha(a^1)^2&\cdots&(a^{n-1})^\alpha(a^{n-1})^2\\
&&\vdots&\\
(a^0)^\alpha (a^0)^{\alpha-1}&(a^1)^\alpha(a^1)^{\alpha-1}&\cdots&(a^{n-1})^\alpha(a^{n-1})^{\alpha-1}
\end{array}\right]\\
=\left[\begin{array}{c}
\bar{G}\\
\bar{G}\Delta
\end{array}
\right],
\end{array}\label{MSR-encoding}
\end{equation}}
and $C$  is the codeword vector with dimension $(\alpha\times n)$.  $\bar{G}$
contains the first $\alpha$ rows in $G$ and $\Delta$ is a diagonal matrix with
$(a^0)^\alpha,\ (a^1)^\alpha,\ (a^2)^\alpha,\ldots,\ (a^{n-1})^\alpha$ as
diagonal elements. Note that if the RS  code is over $GF(2^m)$ for $m\ge \lceil
\log_2 n\alpha\rceil$, then it can be shown that
$(a^0)^\alpha,\ (a^1)^\alpha,\ (a^2)^\alpha,\ldots,\ (a^{n-1})^\alpha$ are all
distinct. After encoding, the $i$th column of $C$ is distributed to storage node $i$
for $1\le i\le n$.
%


\section{Encoding Schemes for Error-Correcting MSR Codes}
\label{SEC:MSR-encoding}
RS codes are known to have very efficient decoding algorithms and exhibit good error correction capability.
From \eqref{MSR-encoding} in Section~\ref{subSEC:MSR}, a generator matrix $G$ for MSR codes needs to satisfy:
\begin{enumerate}
\item $G=\left[\begin{array}{c}
\bar G\\
\bar G\Delta\end{array}\right],$ where $\bar G$ contains the first $\alpha$ rows in $G$ and $\Delta$ is a diagonal matrix with distinct elements in the diagonal.
\item $\bar G$ is a generator matrix of the $[n,\alpha]$ RS code and $G$ is a generator matrix of the $[n,d=2\alpha]$ RS code.
\end{enumerate}
Next, we present a sufficient condition for $\bar G$ and $\Delta$
such that $G$ is a generator matrix of an $[n,d]$ RS code.

\begin{theorem}
\label{thm:MSR-encoding}
Let $\bar G$ be a generator matrix of  the $[n,\alpha]$ RS code $C_{\alpha}$ that is generated by the generator polynomial with roots $a^1,a^2,\ldots,a^{n-\alpha}$. Let the diagonal elements of $\Delta$ be $(a^0)^\alpha$, $(a^1)^\alpha$, $\ldots$, $(a^{n-1})^\alpha$, where $m\ge \lceil \log_2 n\rceil$ and $\gcd(2^m-1,\alpha)=1$. Then $G$ is a generator matrix of $[n,d]$ RS code $C_{d}$ that is generated by the generator polynomial with roots $a^1,a^2,\ldots,a^{n-d}$.
\end{theorem}

\begin{proof}
We need to show that each row of $\bar G\Delta$ is a codeword of $C_{d}$, and
all rows in $G$ are linearly independent. Let $\ccc=(c_0,c_1,\ldots,c_{n-1})$
be any row in $\bar G$. Then the polynomial representation of $\ccc\Delta$ is
\begin{equation}\label{cdelta}
\sum_{i=0}^{n-1}c_i(a^i)^\alpha x^i=\sum_{i=0}^{n-1}c_i(a^\alpha x)^i~.\end{equation}
Since $\ccc\in C_{\alpha}$, $\ccc$ has roots $a^1,a^2,\ldots,a^{n-\alpha}$. Then it is easy to see that~\eqref{cdelta} has roots $a^{-\alpha+1}$, $a^{-\alpha+2}$,$\ldots$, $a^{n-2\alpha}$ that clearly contain $a^1,a^2,\ldots,a^{n-2\alpha}$. Hence, $\ccc\Delta\in C_{d}$.

In order to show that all rows in $G$ are linearly independent, it is sufficient to show that $\ccc\Delta\not\in C_{\alpha}$ for all nonzero $\ccc\in
C_{\alpha}$.
Assume that $\ccc\Delta\in C_{\alpha}$. Then $\sum_{i=0}^{n-1}c_i(a^\alpha
x)^i$ must have roots $a^1,a^2,\ldots,a^{n-\alpha}$. It follows that $c(x)$
must have $a^{\alpha+1},a^{\alpha+2},\ldots,a^{n}$ as roots. Recall that $c(x)$
also has roots $a^1,a^2,\ldots,a^{n-\alpha}$. Since $n-1\ge d=2\alpha$, we have
$n-\alpha\ge \alpha+1$. Hence, $c(x)$ has $n$ distinct roots of
$a^1,a^2,\ldots,a^{n}$. This is impossible since the degree of $c(x)$ is at
most $n-1$. Thus, $\ccc\Delta\not\in C_{\alpha}$. 
\end{proof}

One advantage of the proposed scheme is that it can now operate on a smaller
finite field than that of the scheme in~\cite{HAN12-INFOCOM}.  Another advantage is that one can choose $\bar{G}$ (and
$\Delta$ accordingly) freely as long as it is the generation matrix of an $[n,
\alpha]$ RS code. In particular, as discussed in Section~\ref{SEC:Intro}, to minimize update complexity, it is desirable 
to choose a generator matrix where the row with the maximum Hamming weight has the
least number of nonzero elements. Next, we
present a least-update-complexity generator matrix that satisfies~\eqref{MSR-encoding}.

\begin{corollary}
\label{col:G}
Let  $\Delta$ be the one given in Theorem~\ref{thm:MSR-encoding}. Let $\bar G$
be the generator matrix of  a systematic  $[n,\alpha]$ RS code, namely,
$${\bar G}=\left[D|I\right]$$
where
{\footnotesize  \begin{eqnarray}
D&=&\left[\begin{array}{ccccc}
b_{00}&b_{01}&b_{02}&\cdots&b_{0(n-\alpha-1)}\\
b_{10}&b_{11}&b_{12}&\cdots&b_{1(n-\alpha-1)}\\
b_{20}&b_{21}&b_{22}&\cdots&b_{2(n-\alpha-1)}\\
\vdots&&\vdots&&\vdots\\
b_{(\alpha-1)0}&b_{(\alpha-1)1}&b_{(\alpha-1)2}&\cdots&b_{(\alpha-1)(n-\alpha-1)}
\end{array}\right]~,\label{MSR-G-S}\end{eqnarray}}
$I$ is the $(\alpha\times\alpha)$ identity matrix, and
$$x^{n-\alpha+i}=u_i(x)g(x)+b_i(x)\mbox{ for } 0\le i\le \alpha-1~.$$
Then,  $G=\left[\begin{array}{c}
\bar G\\
\bar G\Delta\end{array}\right]$ is a least-update-complexity generator matrix.
\end{corollary}
\begin{proof}
The result holds since each row of $\bar G$ is a nonzero codeword with the minimum Hamming weight $n-\alpha+1$.
\end{proof}

\section{Efficient Decoding Scheme for error-correcting MSR Codes}
\label{SEC:decoding}
Unlike the decoding scheme in~\cite{HAN12-INFOCOM} that uses $[n,d]$ RS
code, we propose to use the subcode of the $[n,d]$ RS code, the
$[n,\alpha=k-1]$ RS code generated by $\bar{G}$, to perform the data
reconstruction. The advantage of using  the $[n,k-1]$ RS code is two-fold.
First,  its error correction capability is higher (namely, it can tolerate
$\lfloor\frac{n-k+1}{2}\rfloor$ instead of $\lfloor\frac{n-d}{2}\rfloor$
errors). Second, it only requires the access of two additional storage nodes
(as opposed to $d-k+2=k$ nodes) for the first error to correct.  

Without loss of generality, we assume that the data collector retrieves encoded
symbols from $k+2v$ ($v\ge 0$) storage nodes, $j_0,j_1,\ldots,j_{k+2v-1}$.  We
also assume that  there are $v$ storage nodes  whose received symbols are
erroneous.  The stored information of the $k+2v$ storage nodes are collected as
the  $k+2v$ columns in $Y_{\alpha\times (k+2v)}$. The $k+2v$ columns of $G$
corresponding to storage nodes $j_0,j_1,\ldots,j_{k+2v-1}$ are denoted as the
columns of $G_{k+2v}$.  First, we discuss data reconstruction when $v=0$. The
decoding procedure is similar to that in~\cite{RAS11}.

\paragraph*{No Error}
In this case, $v=0$ and there is no error in  $Y$.  Then, 
\begin{eqnarray}
Y&=&UG_{k}\nonumber\\
&=&[Z_1Z_2]\left[\begin{array}{c}\bar{G}_{k}\nonumber\\
\bar{G}_{k}\Delta\end{array}\right]\\&=&[Z_1\bar{G}_{k}+Z_2\bar{G}_{k}\Delta]~.\label{UG-no-error}
\end{eqnarray}
Multiplying $\bar{G}_{k}^T$ and $Y$ in \eqref{UG-no-error}, we have~\cite{RAS11},
{\small \begin{eqnarray}
\bar{G}_{k}^TY&=&\bar{G}_{k}^TUG_{k}\nonumber\\
&=&[\bar{G}_{k}^TZ_1\bar{G}_{k}+\bar{G}_{k}^TZ_2\bar{G}_{k}\Delta]\nonumber\\
&=&P+Q\Delta~.\label{PQ-no-error}
\end{eqnarray}}

Since $Z_1$ and $Z_2$  are symmetric, $P$ and $Q$ are symmetric as well. The
$(i,j)$th element of $P+Q\Delta$, $1\le i,j\le k$ and $i\neq j$,  is
\begin{eqnarray}
p_{ij}+q_{ij}a^{(j-1)\alpha}~,\label{pq-ij}
\end{eqnarray}
and the $(j,i)$th element is given by
\begin{eqnarray}
p_{ji}+q_{ji}a^{(i-1)\alpha}~.\label{pq-ji}
\end{eqnarray}
Since $a^{(j-1)\alpha}\neq a^{(i-1)\alpha}$ for all $i\neq j$, $p_{ij}=p_{ji}$,
and $q_{ij}=q_{ji}$, combining \eqref{pq-ij} and \eqref{pq-ji}, the values of
$p_{ij}$ and $q_{ij}$ can be obtained. Note that we only obtain $k-1$ values
for each row of $P$ and $Q$ since no elements in the diagonal of $P$ or $Q$ are
obtained.

To decode $P$, recall that $P=\bar{G}_{k}^TZ_1\bar{G}_{k}$. $P$ can
be treated as a portion of the codeword vector, $\bar{G}_{k}^TZ_1\bar{G}$. By
the construction of $\bar{G}$, it is easy to see that $\bar{G}$ is a generator
matrix of the $[n,k-1]$ RS code. Hence, each row in the matrix
$\bar{G}_{k}^TZ_1\bar{G}$ is a codeword. Since we have known $k-1$ components
in each row of $P$, it is possible to decode $\bar{G}_{k}^TZ_1\bar{G}$ by the
error-and-erasure decoder of the $[n,k-1]$ RS code.\footnote{ The
error-and-erasure decoder of an $[n,k-1]$ RS code can successfully decode a
received vector if $s+2v<n-k+2$, where $s$ is the erasure (no symbol)
positions, $v$ is the number of errors in the received portion of the received
vector, and $n-k+2$ is the minimum Hamming distance of the $[n,k-1]$ RS code.}

 Since one cannot locate any erroneous position from the decoded rows of $P$, the decoded $\alpha$ codewords are accepted as $\bar{G}_{k}^TZ_1\bar{G}$. By collecting the last $\alpha$ columns of $\bar{G}$ as $\bar{G}_\alpha$ to find its inverse (here it is an identity matrix), one can recover $\bar{G}_{k}^TZ_1$ from $\bar{G}_{k}^TZ_1\bar{G}_k$. Note that $\alpha=k-1$.
Since any $\alpha$ rows in $\bar{G}_{k}^T$ are independent and thus invertible, we can pick any $\alpha$ of them to recover $Z_1$.
$Z_2$ can be obtained similarly by $Q$.

\paragraph*{Multiple  Errors}
Before presenting the proposed decoding algorithm, we first prove that a
decoding procedure can always successfully decode $Z_1$ and $Z_2$ if $v\le
\lfloor\frac{n-k+1}{2}\rfloor$ and all storage nodes are accessed. Due to space
limitation, all proofs are omitted in this section.

Assume the storage nodes with errors correspond to the $\ell_0$th,
$\ell_1$th, $\ldots$, $\ell_{v-1}$th columns in the received matrix
$Y_{\alpha\times n}$. Then,
 \begin{eqnarray}
&&\bar{G}^TY_{\alpha\times n}\nonumber\\
&=&\bar{G}^TUG+\bar{G}^TE\nonumber\\
&=&\bar{G}^T[Z_1Z_2]\left[\begin{array}{c}\bar{G}\nonumber\\
\bar{G}\Delta\end{array}\right]+\bar{G}^TE\\&=&[\bar{G}^TZ_1\bar{G}+\bar{G}^TZ_2\bar{G}\Delta]+\bar{G}^TE~,\label{UG-error}
\end{eqnarray}
where
{\small $$E=\left[\0_{\alpha\times (\ell_{0}-1)}|\eee^T_{\ell_0}|\0_{\alpha\times (\ell_{1}-\ell_{0}-1)}|\cdots|\eee^T_{\ell_{v-1}}|\0_{\alpha\times (n-\ell_{v-1})}\right]~.$$}
\begin{lemma}
\label{lemma}
There are at least $n-k+2$ errors in each of the $\ell_0$th, $\ell_1$th, $\ldots$, $\ell_{v-1}$th columns of $\bar{G}^TY_{\alpha\times n}$.
\end{lemma}
We next have the main theorem to perform data reconstruction.
\begin{theorem}
\label{thm:main}
Let $\bar{G}^TY_{\alpha\times n}=\tilde{P}+\tilde{Q}\Delta$. Furthermore, let
$\hat{P}$ be the corresponding portion of decoded codeword vector to
$\tilde{P}$ and  $E_P=\hat{P}\oplus \tilde{P}$ be the error pattern vector.
Assume that the data collector accesses all storage nodes and there are $v$,  $1\le
v\le \lfloor\frac{n-k+1}{2}\rfloor$, of them with errors. Then, there are at
least $n-k+2-v$ nonzero elements in $\ell_{j}$th column of $E_P$,  $0\le j\le
v-1$, and at most $v$ nonzero elements in the rest of columns of $E_P$.
\end{theorem}

The above theorem allows us to design a decoding algorithm that can correct up
to $\lfloor\frac{n-k+1}{2}\rfloor$ errors.\footnote{ In constructing
$\tilde{P}$ we only get $n-1$ values (excluding the diagonal). Since the
minimum Hamming distance of an $[n,k-1]$ RS code is $n-k+2$, the
error-and-erasure decoding can only correct up to
$\lfloor\frac{n-1-k+2}{2}\rfloor$ errors.}
In particular, we need to examine the erroneous positions in $\bar{G}_{k+3}^TE$.  Since $1\le v\le
\lfloor\frac{n-k+1}{2}\rfloor$, we have $n-k+2-v\ge
\lfloor\frac{n-k+1}{2}\rfloor+1>v$. Thus, the way to locate all erroneous
columns in $\tilde{P}$ is to find out all columns in $E_P$ where the number of
nonzero elements in them are greater than or equal to
$\lfloor\frac{n-k+1}{2}\rfloor+1$. After we locate all erroneous columns we can
follow a procedure similar to that given in the no error case to recover $Z_1$ from
$\hat{P}$.

The above decoding procedure guarantees to recover $Z_1$ when all $n$ storage
nodes are accessed. However, it is not very efficient in terms of bandwidth usage.
Next, we present a progressive decoding version of the proposed algorithm that
only accesses enough extra nodes when necessary. Before presenting it, we
need the following corollary.
\begin{corollary}
\label{cor}
Consider that one accesses $k+2v$ storage nodes, among which $v$ nodes are
erroneous and  $1\le v\le \lfloor\frac{n-k+1}{2}\rfloor$. There are at least
$v+2$ nonzero elements in $\ell_{J}$th column of $E_P$,  $0\le j\le v-1$, and
at most $v$ among the remaining columns of $E_P$.
\end{corollary}

Based on Corollary~\ref{cor}, we can design a progressive decoding
algorithm~\cite{HAN12} that retrieve extra data from remaining storage nodes
when necessary. To handle Byzantine fault tolerance, it is necessary to perform
integrity check after the original data is reconstructed.  Two verification
mechanisms have been suggested in~\cite{HAN12-INFOCOM}: cyclic redundancy check
(CRC) and cryptographic hash function. Both mechanisms introduce redundancy to
the original data before they are encoded and are suitable to be used in
combination with the decoding algorithm.

The progressive decoding algorithm starts from accessing $k$ storage nodes. 
Error-and-erasure decoding succeeds only when there is no
error. If the integrity check passes, then the data collector recovers the
original data. If the decoding procedure  fails or the integrity check fails,
then the data collector retrieves  two more  blocks of data from the remaining storage
nodes. Since the data collector has $k+2$ blocks of  data, the error-and-erasure
decoding can correctly recover the original data if there is only one erroneous
storage node among the $k+1$ nodes accessed. If the integrity check passes,
then the data collector recovers the original data. If the decoding procedure
fails or the integrity check fails, then the data collector retrieves two more
blocks of  data from the remaining storage nodes. The data collector repeats the same
procedure until it recovers the original data or runs out of the storage nodes.
The detailed decoding procedure is summarized in
Algorithm~\ref{algo:reconstruction-MSR}.

\begin{algorithm}[h]
\Begin {
$v=0$; $j=k$;\\
The data collector randomly chooses $k$ storage nodes and retrieves encoded data,
$Y_{\alpha\times j}$;\\
\While {$v \le \lfloor\frac{n-k+1}{2}\rfloor$} {
Collect the $j$ columns of $\bar G$ corresponding to accessed storage nodes as  $\bar G_{j}$;\\
Calculate $\bar G_{j}^TY_{\alpha\times j}$;\\
Construct $\tilde{P}$  and  $\tilde{Q}$ by using \eqref{pq-ij} and \eqref{pq-ji};\\
Perform progressive error-and-erasure decoding on each row in $\tilde{P}$  to obtain $\hat{P}$;\\
Locate  erroneous columns in $\hat{P}$  by searching for columns of them with at least $v+2$ errors;
 assume that $\ell_e$ columns found in the previous action;\\
Locate columns in $\hat{P}$ with at most $v$ errors; assume that $\ell_c$ columns found in the previous action;\\
\If{($\ell_e=v$ and $\ell_c=k+v$)} {
Copy  the $\ell_e$ erronous columns of $\hat{P}$ to their corresponding rows to make $\hat{P}$  a symmetric matrix;\\
Collect any $\alpha$ columns in the above $\ell_c$ columns of $\hat{P}$ as $\hat{P}_\alpha$ and find its corresponding $\bar{G}_\alpha$;\\
Multiply  the inverse of $\bar{G}_\alpha$ to $\hat{P}_\alpha$ to recover $\bar{G}_{j}^TZ_1$;\\
Recover $Z_1$ by the inverse of any $\alpha$ rows of $\bar{G}_{j}^T$;\\
Recover $Z_2$ from $\tilde{Q}$ by the same procedure; Recover $\tilde{\m}$ from $Z_1$ and $Z_2$;\\
\If{ integrity-check($\tilde{\m}$) = SUCCESS} {
\Return $\tilde{\m}$;
} }
$j \leftarrow j+2$;\\
Retrieve $2$ more encoded data from remaining storage nodes and merge them into $Y_{\alpha\times j}$; $v\leftarrow v+1$;

}
\Return FAIL;
}
\caption{Decoding of MSR Codes Based on  $(n,k-1)$ RS Code for Data Reconstruction}
\label{algo:reconstruction-MSR}
\end{algorithm}

\begin{figure}
\centering
\includegraphics[width=8cm]{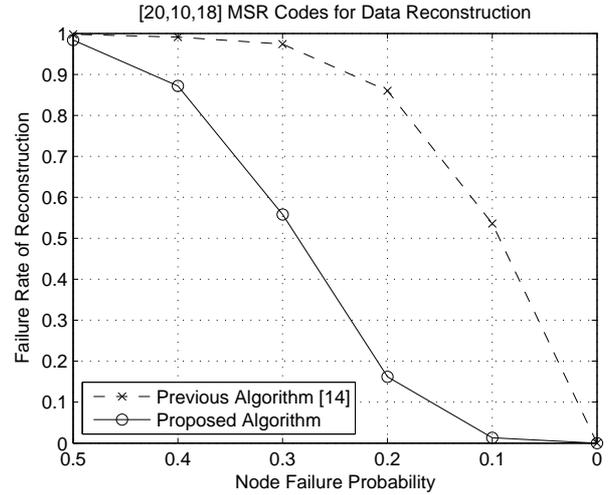}
\caption{Failure-rate comparison between the  previous algorithm in~\cite{HAN12-INFOCOM} and the proposed algorithm for $[20,10,18]$ MSR code} \label{fig:fig1}
\end{figure}


The proposed data reconstruction algorithm for MSR codes is evaluated by Monte Carlo simulations. It is compared with the previous data reconstruction algorithms in~\cite{HAN12-INFOCOM}. Each
data point is generated from $10^3$ simulation results. Storage nodes may fail
arbitrarily with the Byzantine failure probability ranging from $0$ to $0.5$. $[n,k,d]$ and $m$ are chosen to be $[20,10,18]$ and $5$ respectively. Figure~\ref{fig:fig1} shows that the proposed algorithm can successfully reconstruct the data with much higher probability than the  one presented in~\cite{HAN12-INFOCOM} at the same node failure probability. For example, at the node failure probability of $0.1$, data for about $1$ percent of node failure patterns cannot be reconstructed using the proposed algorithm. On the other hand, data for over $50$ percents of node failure patterns cannot be reconstructed using the previous algorithm in~\cite{HAN12-INFOCOM}. The advantage of the proposed algorithm is also overwhelming in the average number of accessed nodes for data reconstruction. Due to space limitation, the simulation results are omitted.

\section{Conclusion}
\label{SEC:conclude}
In this work we proposed a new encoding scheme for the $[n,2\alpha]$
error-correcting MSR codes from the generator matrix of any $[n, \alpha]$ RS
codes. It generalizes the previously proposed MSR codes in~\cite{HAN12-INFOCOM}
and has several salient advantages. It allows the construction of
least-update-complexity codes with a properly chosen systematic generator
matrix. More importantly, the decoding scheme leads to an efficient decoding
scheme that can tolerate more errors at the storage nodes, and access
additional storage nodes only when necessary. A progressive decoding scheme was
thereby devised with low communication overhead. 

Possible future work includes extension of the encoding and decoding schemes to
MBR points, and the study of encoding schemes with optimal update complexity
and good regenerating capability. 
\section*{Acknowledgment}
This work was supported in part by  CASE: The Center for Advanced Systems and Engineering, a NYSTAR center for advanced technology at Syracuse University; the National Science of Council (NSC) of
Taiwan under grants no. 99-2221-E-011-158-MY3 and NSC 101-2221-E-011-069-MY3;
US National Science Foundation under grant no. CNS-1117560 and McMaster
University new faculty startup fund. 
\bibliographystyle{IEEEtran}
\bibliography{IEEEabrv,network,nfs}

\end{document}